\documentclass[
  onecolumn,
  superscriptaddress,
  a4paper,
  11pt,
  accepted=2020-09-01]{quantumarticle}

\pdfoutput=1

\usepackage[numbers,sort&compress]{natbib}

\usepackage[english]{babel}
\usepackage{amsmath,amsfonts,amsthm,mathrsfs,mathpazo,xspace,hyperref,graphicx}
\usepackage{endnotes}
\usepackage{color}
\usepackage{bm}
\usepackage{amssymb,latexsym}

\newtheorem{theorem}{Theorem}

\newtheorem{lemma}[theorem]{Lemma}

\newtheorem{fact}[theorem]{Fact}

\theoremstyle{remark}
\newtheorem{remark}[theorem]{Remark}

\theoremstyle{definition}
\newtheorem{definition}[theorem]{Definition}

\newcommand{\beq}{\begin{eqnarray}}
\newcommand{\eeq}{\end{eqnarray}}

\newcommand{\ket}[1]{|#1\rangle}
\newcommand{\bra}[1]{\langle#1|}
\newcommand{\proj}[1]{\ket{#1}\!\bra{#1}}

\newcommand{\reg}[1]{{\textsf{#1}}}

\newcommand{\C}{\ensuremath{\mathbb{C}}}

\newcommand{\mH}{\mathcal{H}}

\newcommand{\mQ}{\mathcal{Q}}
\newcommand{\mP}{\mathcal{P}}

\newcommand{\setft}[1]{\mathrm{#1}}

\newcommand{\Lin}{\setft{L}}

\newcommand{\eps}{\varepsilon}

\newcommand{\GHZ}{\textsc{GHZ}}
\newcommand{\MS}{\textsc{MS}}
\newcommand{\EPR}{\textsc{EPR}}
\newcommand{\pauli}{\textsc{P3}}

\newcommand{\game}{\textsc{3EMB}}

\newif\ifnotes\notesfalse

\ifnotes
\usepackage{color}
\definecolor{mygrey}{gray}{0.50}
\newcommand{\notename}[2]{{\textcolor{mygrey}{\footnotesize{\bf (#1:} {#2}{\bf ) }}}}

\newcommand{\pnote}[1]{{\endnote{#1}}}

\else

\newcommand{\notename}[2]{{}}

\newcommand{\pnote}[1]{}

\fi

\newcommand{\tnote}[1]{\textcolor{magenta}{\small {\textbf{(Thomas:} #1\textbf{)
      }}}}

\def\<{\langle}
\def\>{\rangle}

\usepackage[margin=1in]{geometry}
\newcommand{\comment}[1]{}

\begin{document}

\title{A three-player coherent state embezzlement game}
\author{Zhengfeng Ji}
\affiliation{Centre for Quantum Software and Information, University of Technology Sydney, Australia}
\author{Debbie Leung}
\affiliation{University of Waterloo and the Perimeter Institute, Canada.  Email: \texttt{wcleung@uwaterloo.ca}}
\author{Thomas Vidick}
\affiliation{California Institute of Technology, USA. Email: \texttt{vidick@cms.caltech.edu}}
\date{}
\maketitle

\begin{abstract}
We introduce a three-player nonlocal game, with a finite number of classical questions and answers, such that the optimal success probability of $1$ in the game can only be achieved in the limit of strategies using arbitrarily high-dimensional entangled states. Precisely, there exists a constant $0 <c\leq 1$ such that to succeed with probability $1-\eps$ in the game it is necessary to use an entangled state of at least $\Omega(\eps^{-c})$ qubits, and it is sufficient to use a state of at most $O(\eps^{-1})$ qubits.

The game is based on the coherent state exchange game of Leung et al.\ (CJTCS 2013). In our game, the task of the quantum verifier is delegated to a third player by a classical referee. Our results complement those of Slofstra (arXiv:1703.08618) and Dykema et al.\ (arXiv:1709.05032), who obtained two-player games with similar (though quantitatively weaker) properties based on the representation theory of finitely presented groups and $C^*$-algebras respectively.
\end{abstract}

A nonlocal game \cite{cleve2004consequences} consists of a one-round
interaction between a trusted referee and multiple (say, $k \geq 2$)
spatially isolated players.  In the game, the referee draws a query,
which is a $k$-tuple of questions, according to a publicly known
distribution and sends one question to each player.  Upon receiving
its question, each player sends an answer back to the referee.  The
referee decides if the players win or not by checking if the $k$-tuple
of answers satisfies a publicly specified criterion depending only on
the selected query and the answers.  It is always assumed that the
players cooperate to maximize their probability of winning the game.
The players may be allowed to share different types of nonlocal
resources, such as shared randomness or entanglement, before the game
starts, but they cannot communicate during the game.  Despite the fact
that the interaction with the referee is entirely classical, there are
nonlocal games in which the players can win with strictly higher
probability by sharing entanglement as opposed to only classical
shared randomness~\cite{Bell:64a}.

Nonlocal games play a role in many areas and have been studied under
multiple guises.  First, any nonlocal game is a Bell
experiment, where each question is a measurement setting, each answer
is a measurement outcome, and an upper bound on the maximum
probability of winning using shared randomness alone is a Bell
inequality.  This leads to a theoretical framework to ``test
quantumness''~\cite{Bell:64a} that has been experimentally
demonstrated with great success, from pioneering work
in~\cite{aspect1981experimental} to recent loophole-free
violations~\cite{hensen2015loophole,giustina2015significant,shalm2015strong}.
Second, nonlocal games have been actively studied in complexity theory
in relation to the study of interactive proof systems.  Third, some
nonlocal games have a useful \emph{rigidity} property that implies
that there is an essentially unique optimal winning strategy.  Thus
correlations in the answers can be used as a certificate for the
entangled state shared by the parties as well as the measurements
applied.  Such results are very useful in quantum cryptography to
achieve tasks such as device independent quantum key distribution or
verified delegated quantum computation~\cite{reichardt2013classical}.  

An outstanding question in the theory of nonlocal games is the quantification of the amount of entanglement required to achieve, or even approach, optimality. For a long time, there was no explicit nonlocal game known for which any optimal strategy provably required more than one, or at most two, qubits of entanglement per player. More recently, a number of examples of nonlocal games requiring a large amount of entanglement have been found, see for example \cite{BBT11,slofstra2011lower,ji2013binary,manvcinska2014unbounded,chao2016test,coladangelo2017separation}.  However, these games all require an increasing number of questions or answers.  In \cite{pal2010maximal} the authors identified a Bell inequality, the so-called $I_{3322}$ inequality, that can also be formulated as a two-player game with 3 possible questions and 2 possible answers per player, and gave strong numerical evidence that the optimal violation of the inequality (equivalently, the maximum success probability of players sharing entanglement in the associated game) could only be reached in the limit of arbitrarily high-dimensional entanglement. However, an analytical proof of this fact has remained elusive.

In Leung et al.~\cite{leung2013coherent}, motivated by the discovery of ``embezzling states''~\cite{van2003universal} and to gain insights in the amount of entanglement required of optimal strategies in multi-prover interactive proof systems, the authors introduced a game called the ``coherent state exchange game'' in which each player receives a $3$-dimensional system and returns a qubit.  They showed that an optimal success probability of $1$ in this game could only be achieved in the limit of strategies using entangled states of arbitrarily large dimension; moreover, they provided precise trade-offs between success probability and dimension.
The intuition for the game is simple: the players are tasked with coherently transforming a product state to an EPR pair. A simple application of Fannes' inequality~\cite{fannes1973continuity} shows that this can only be accomplished by using an arbitrarily large ``reservoir'' of entanglement. Such ``universal reservoirs'', the embezzlement states introduced in~\cite{van2003universal}, can be used to instantiate arbitrarily close to perfect strategies for the players. 

The game considered in~\cite{leung2013coherent} is not a nonlocal game in the orthodox sense of the term: in the game the referee is required to prepare a (small) entangled state, and exchange quantum states with the players. Was this a ``cheat'' that enabled the result, or a hint that a similar property should be achievable with nonlocal games with a classical referee? In~\cite{regev2015quantum} a step was taken in this direction by adapting the game to one in which questions remain quantum, but answers from the player are classical. \tnote{Added:}In their game, the players are asked to locally distinguish the two states 
\[ \ket{\psi_0} = \frac{1}{\sqrt{2}}\ket{0}\ket{0} + \frac{1}{\sqrt{2}}\ket{\EPR} \quad\text{and}\quad \ket{\psi_1}= \frac{1}{\sqrt{2}}\ket{0}\ket{0} - \frac{1}{\sqrt{2}}\ket{\EPR} \;,\]
where $\ket{\EPR}$ denotes an EPR pair, by each returning a single bit in $\{\pm 1\}$ such that the product of the two bits is $+1$ in case the state sent by the verifier is $\ket{\psi_0}$, and $-1$ in case it is $\ket{\psi_1}$. In~\cite{regev2015quantum} the authors show that this task can only be accomplished perfectly in the limit of infinite-dimensional entanglement. 

A breakthrough came in a sequence of two works by Slofstra~\cite{slofstra2016tsirelson,slofstra2017a}, who introduced completely different techniques, based on the representation theory of finitely presented groups and a ``universal embedding theorem'' to obtain nonlocal games from groups. A consequence of Slofstra's work is the resolution of a decades-old open question on the closure of the set of finite-dimensional quantum correlations, showing that this set is not closed. In particular, there exists a finite game such that the optimum success probability cannot be achieved in any finite dimension, resolving the aforementioned line of questioning in the affirmative.

A different proof for the non-closure of the set of quantum correlations was later obtained by Dykema et al.~\cite{DPP2017}. Although the proof is arguably simpler and more direct (in particular, it yields a two-player game with only $5$ questions per player!), it still relies on rather non-trivial mathematical results establishing the non-existence of non-trivial finite-dimensional representations for certain $C^*$-algebras associated with projections. A drawback of these and Slofstra's methods is that it may not be obvious to formulate the resulting game explicitly, to gain insights on the physical reason why increasing amounts of entanglement can be required to win with higher probabilities, or to obtain good quantitative estimates on the achievable trade-offs between dimension and success probability (though a step in this direction was recently made by making a special case of Slofstra's approach quantitative: see~\cite{slofstra2017b}, on which we comment more below).

\paragraph{Our results.}
In this paper we return to the line of works~\cite{leung2013coherent,regev2015quantum} exploring the properties of quantum embezzlement, and provide a different, arguably more direct and more intuitive construction of a nonlocal game, with classical questions and answers, whose optimal success probability of $1$ can only be achieved in the limit of infinite-dimensional strategies. A benefit is that our construction is fully explicit, and we are able to obtain precise quantitative estimates on the trade-off between dimension and success probability of any strategy. Our analysis shows that any near-optimal
strategy for the game we construct must contain, within itself, the ability to ``embezzle'' an EPR
pair from a product state -- a task that, according to Fannes'
inequality, can only be achieved with arbitrarily high accuracy using a
family of ancilla entangled states that have unbounded entanglement
entropy. The impossibility of perfect embezzlement using
finite-dimensional entanglement thus provides a natural physical basis
for the fact that the optimal success probability of $1$ in our game can only be achieved in the limit of infinite-dimensional strategies.

As already mentioned, our starting point is the two-player
embezzlement game~\cite{leung2013coherent}.
We modify the two-player game with quantum referee into a
three-player game with classical referee by
turning the quantum referee in the original game into a
third player in the new game.  The classical referee
in the new game
classically
``delegates'' to the third player
the preparation of the quantum referee's messages to
the other two players.
The transformation follows a similar spirit as a family of more general
transformations introduced by
Ji~\cite{ji2016classical,ji2017compression}.
It is not clear if the techniques
from~\cite{ji2016classical,ji2017compression} could work here as a
black-box. In addition, even if the constructions proposed in those works did lead to nonlocal games with the desired
properties, the games would have at least four extra players, and the
analysis would be non-trivial.
Here, we give a more direct construction, with a simple analysis, that only
requires a single additional player.

Our game, called $\game$, is described
in Figure \ref{fig:game-n}. It satisfies the properties described in
Theorem \ref{thm:main}.

\begin{theorem}\label{thm:main}
There exists a three-player game with the following properties:
\begin{itemize}
\item There are 12 possible questions to each player. One player replies with 3 bits and the other two each reply with 2 bits.
\item For any $\eps >0$ there is $d=O(\eps^{-1})$ and a strategy for the players that succeeds with probability $1-\eps$ using an entangled state with local dimension $8$ for the first player and $2^{d+3}$ for each of the other two players.
\item There is a constant $c>0$ such that for any $\eps >0$, any strategy for the players that succeeds with probability at least $1-\eps$ in the game must use an entangled state of local dimension at least $2^{\Omega(\eps^{-c})}$ for two of the players.  
\end{itemize}
\end{theorem}

Our game is smaller than the game from~\cite{slofstra2017b}, but larger than the one from~\cite{DPP2017}; in addition, it requires three players, instead of two for both of these results. Quantitatively, the trade-off between dimension and success probability we obtain is exponentially stronger than the one obtained in~\cite{slofstra2017b}. (An exponential trade-off of the kind we obtain has long been known for \emph{families} of games, but of course the point of our result is that the trade-off is demonstrated for a single, finite game.)

\paragraph{Subsequent work.}  
Subsequent to the initial submission of the current paper to the arXiv,
several related results have been reported. 

Slofstra \cite{slofstra2018} found a two-player nonlocal game such
that playing the game with success probability $1-\epsilon$ requires
the players to use local Hilbert spaces of dimension at least $c'
\exp(c/\epsilon^\alpha)$ for some positive constants $c,c',\alpha$. 

Coladangelo and Stark \cite{coladangelostark2018} constructed a
correlation for $2$ players in which one player receives $4$ questions
and the other receives $5$ questions, and each has $3$ answers to
choose from.
The correlation cannot be attained in finite dimensions, but there is
a Cauchy sequence of finite dimensional entangled states that
enables the correlation to be approximated to arbitrary precision.
The proofs in \cite{coladangelostark2018} are elementary; however, the
authors do not provide quantitative bounds for the entanglement
dimension required by a certain approximation.

Furthermore, inspired by the approach in this paper to turn the
two-player embezzlement game in \cite{leung2013coherent} to a nonlocal
game, Coladangelo \cite{coladangelo2019} obtained an expliciit Bell
inequality for two players, that can be rewritten as a nonlocal game
with 6 questions and 3 answers for each player, with a similar
trade-off between the probability of winning and entanglement
dimension.  This addresses our open problem in an earlier draft
whether a two-player nonlocal game can be obtained in the current
approach with exponential scaling.  The games in this paper and in
\cite{leung2013coherent,coladangelo2019} feature three different
mechanisms to enforce the coherent transformation of a product state to
an EPR pair.

More recently, Fu~\cite{Fu2019}, building on the embedding procedure developed in \cite{slofstra2017a}, found a family of constant-sized correlations that can only 
be approximated up to within sufficiently small error if the parties share 
maximally entangled states with increasing dimension.

\paragraph{Discussion.} It remains an outstanding open question to determine the size of the smallest game such that the optimal success probability in the game can only be achieved in the limit of infinite-dimensional strategies. All games for which such a result has been shown so far have quantum value $1$ (also called ``pseudo-telepathy'' games), whereas in the case of the $I_{3322}$ inequality, the (still conjectural) separation is for a game with quantum value strictly less than $1$. It is interesting to explore what features of entanglement cannot be demonstrated in pseudo-telepathy games.

Due to the fact that optimal strategies for the players in our game are required to perform coherent state embezzlement, the results of~\cite{cleve2017perfect} imply that there is no perfect infinite-dimensional strategy in the tensor product model, but there is one in the commuting-operator model.\footnote{We thank Laura Man\v{c}inska for pointing out this consequence.} As a consequence our game is not a candidate for separating the sets $C_q$ and $C_{qs}$ of correlations achievable using finite-dimensional and infinite-dimensional strategies in the tensor product model respectively; this separation, sometimes referred to as ``Tsirelson's problem'', was subsequently proved in \cite{coladangelostark2018}. 

There are reasons to believe that the exponential trade-off between entanglement dimension and success probability demonstrated by our construction may be optimal. Indeed, even if one allows games whose size grows with $\eps^{-1}$ (equivalently, if one restricts to ``not too small'' values of $\eps$), the best scaling known remains exponential (see e.g.~\cite{ostrev2016entanglement} for the best known in the case of XOR games, or the aforementioned recent work of Slofstra~\cite{slofstra2017a} that provides a constant-sized game with such scaling). However, no upper bounds are known.

\paragraph{Organization.}
The construction of the game, and its analysis, combines known rigidity results for the $\GHZ$ game and the Magic Square games. These games are described in Section~\ref{sec:tests} and combined in Section~\ref{sec:pauli}. In Section~\ref{sec:game} we introduce the game $\game$, give intuition for the construction, and prove Theorem~\ref{thm:main}.

\section{Preliminaries}

\subsection{Notation}

$\mH$ denotes a finite-dimensional Hilbert space, and $\Lin(\mH)$ the linear operators on $\mH$. We use indices $\mH_\reg{A}$, $\mH_{\reg{B}}$, etc., to index different spaces.
We write
\begin{equation}\label{eq:pauli-matrix}
\sigma_i = \begin{pmatrix} 1 & 0 \\ 0 & 1 \end{pmatrix},\qquad \sigma_x = \begin{pmatrix} 0 & 1 \\ 1 & 0 \end{pmatrix},\qquad \sigma_y = \begin{pmatrix} 0 & -i \\ i & 0 \end{pmatrix}\qquad\text{and}\qquad \sigma_z = \begin{pmatrix} 1 & 0 \\ 0 & -1\end{pmatrix}
\end{equation}
 for the standard single-qubit Pauli observables on $\C^2$. We sometimes use an additional subscript, $\sigma_{w,\reg{R}}$ for $w\in\{i,x,y,z\}$, to clarify the space on which a Pauli operator acts: $\sigma_{w,\reg{R}}$ acts on $\mH_{\reg{R}} \simeq (\C^2)_{\reg{R}}$. We write $\ket{\EPR} = \frac{1}{\sqrt{2}}\ket{00} + \frac{1}{\sqrt{2}} \ket{11}$ for the EPR pair and $\ket{\GHZ} = \frac{1}{\sqrt{2}} \ket{000} + \frac{1}{\sqrt{2}} \ket{111}$ for the $3$-qubit $\GHZ$ state.

We use the following useful piece of notation:

\begin{definition}\label{def:isometry}
For finite-dimensional Hilbert spaces $\mH_{\reg{A}}$, $\mH_\reg{B}$ and $\mH_{\reg{A}'}$, $\delta>0$, and operators $R \in\Lin(\mH_{\reg{A}})$ and $S\in\Lin(\mH_{\reg{A}'})$ we say that $R$ and $S$ are $\delta$-isometric with respect to $\ket{\psi} \in \mH_{\reg{A}} \otimes \mH_{\reg{B}}$, and write $R\simeq_\delta S$, if there exists an isometry $V:\mH_{\reg{A}}\to\mH_{\reg{A}'}$ such that
$$\big\|( R-V^\dagger SV)\otimes I_{\reg{B}} \ket{\psi}\big\|=O(\delta).$$
If $V$ is the identity, then we further say that $R$ and $S$ are $\delta$-equivalent, and write $R\approx_\delta S$ for $\| ( R- S) \otimes I_{\reg{B}} \ket{\psi}\|=O(\delta)$.

Analogously, for a state $\ket{\phi}$ on $\mH_{A'} \otimes \mH_{B'}$ we write $\ket{\psi} \simeq_\delta \ket{\phi}$ when there exists isometries $V_A:\mH_{\reg{A}}\to\mH_{\reg{A}'}$ and $V_B:\mH_{\reg{B}}\to\mH_{\reg{B}'}$ such that $\|\ket{\phi} - V_A \otimes V_B \ket{\psi}\| = O(\delta)$, and $\ket{\psi} \approx_\delta \ket{\phi}$ whenever $V_A$ and $V_B$ are the identity.
\end{definition}

The notation $R\simeq_\delta S$ carries some ambiguity, as it does not specify the state $\ket{\psi}$. The latter should always be clear from the context: we will often simply write that $R$ and $S$ are $\delta$-isometric, without explicitly specifying $\ket{\psi}$ or the isometry. The relation is transitive, but not reflexive: the operator on the right will always act on a space of dimension at least as large as that on which the operator on the left acts. The notion of $\delta$-equivalence is both transitive and reflexive, and we will use it as a measure of distance on linear operators.

\subsection{Elementary tests}
\label{sec:tests}

We use the language of tests to describe elementary building blocks used in the construction of our game. A test is a protocol describing an interaction between a trusted verifier and multiple untrusted players. In the test, the verifier selects a question for each player, according to a publicly known distribution. The (ordered) tuple of questions selected by the verifier is called a query. Upon receiving its question, each player has to provide an answer to the verifier. Finally, the verifier decides to accept (in which case we say that the players pass the test) or reject (the players fail), by evaluating a publicly known predicate on the query and the tuple of answers.

We recall two well-known tests. The first is a test such that any players that pass the test with probability close to $1$ must use a shared entangled state that is isometric to a $\GHZ$ state (we say the test ``self-tests'' the $\GHZ$ state). The second is the Magic Square game, which self-tests two EPR pairs, as well as Pauli $\sigma_x$ and $\sigma_z$ measurements on that state.

\begin{theorem}[GHZ test, Proposition 4 in~\cite{miller2012optimal}]\label{thm:ghz}
There exists a three-player test
$\GHZ$
with the following properties.
\begin{enumerate}
\item The marginal distribution on questions to each player is uniform over $\{x,y\}$;
\item Each player replies with a single bit in $\{\pm 1\}$;
\item For any pair of anti-commuting binary observables $X,Y$ on $\mH_i$, for each player $i\in\{1,2,3\}$,
 there is a strategy for the players that succeeds with probability $1$ and only requires that each player measures her share of an eigenvalue-$1$ eigenstate of the operator
\begin{equation}\label{eq:ghz-op}
G(X,Y)\,=\,\frac{1}{4}\big(X\otimes X \otimes X - Y\otimes Y \otimes X - X \otimes Y \otimes Y - Y \otimes X \otimes Y\big)
\end{equation}
in $\Lin(\mH_1\otimes\mH_2\otimes\mH_3)$ using the binary observable indicated by her question;
\item For any $\eps\geq 0$ there is $\delta = O(\sqrt{\eps})$ such that for any strategy with success probability at least $1-\eps$, there are local isometries on $\mH_i$, $i\in\{1,2,3\}$, such that, under the isometries, $\ket{\psi} \simeq_\delta \ket{\GHZ}_{123}\ket{\psi'}$, for some state $\ket{\psi'}$, and a player's observable $W$ on question $w\in\{x,y\}$ satisfies $W\simeq_\delta \sigma_{w}$, the Pauli observable acting on the player's share of $\ket{\GHZ}_{123}$.
\end{enumerate}
\end{theorem}

One realization of a game satisfying the above theorem can be stated
as follows.  The referee draws a random {query} from the set
$\{(x,x,x),(y,y,x),(y,x,y),(x,y,y)\}$ and sends the $i$-th symbol
in the query to the $i$-th player as her {question}.
The players win if their
answers multiply to $1$ if the query is $(x,x,x)$ and $-1$ otherwise.

In the $\GHZ$ game, we work with the $\GHZ$ state and observables
$\sigma_x$ and $\sigma_y$.
Yet, when designing the main nonlocal game introduced later, it will
be important for us to work with the $\GHZ$ state, and simultaneously have access to
$\sigma_x$ and $\sigma_z$ observables by rigidity.
For this reason, in Section~\ref{sec:pauli} we introduce a game that requires two $\GHZ$ states
and uses a construction motivated by the Magic Square game $\MS$.
In the following theorem, we recall some properties of the Magic Square
game, a two-player game that self-tests two EPR pairs.

\vspace*{1ex}


\begin{theorem}[Magic Square test, Theorem 5.9 in~\cite{coladangelo2017robust}]\label{thm:ms-rigid}
There exists a two-player test $\MS$ with the following properties:
\begin{enumerate}
\item Queries $(q_1,q_2)$ in the game, where for $i\in\{1,2\}$ question $q_i$ goes to
  the $i$-th player, are drawn from $\mQ \times \mQ$, where
\begin{equation}\label{eq:def-e-p}
  \mQ \,= \,\{c_1,c_2,c_3,r_1,r_2,r_3\}.
\end{equation}
The marginal distribution on questions to each player is uniform over $\mQ$;
\item Each player replies with $2$ bits in $\{\pm 1\}^2$;
\item For each player labeled by $i\in\{1,2\}$, let $\mH_i$ denote the Hilbert space associated with player $i$'s local system.  For any two commuting pairs of anti-commuting binary observables $(X^i_1,Z^i_1)$ and $(X^i_2,Z^i_2)$ acting on $\mH_i$, there is a strategy for the players that succeeds with probability $1$ and only requires the measurement of observables obtained as the product of $X^i_1,Z^i_1,X^i_2,Z^i_2$ on an eigenvalue-$1$ eigenstate of the operator
$$\MS(X,Z)\,=\,\frac{1}{2}\big(X^1_1\otimes X^2_1 + Z^1_1\otimes Z^2_1\big)\cdot \frac{1}{2}\big(X^1_2\otimes X^2_2 + Z^1_2\otimes Z^2_2\big)\,\in\,\Lin(\mH_1\otimes \mH_2)\;;\footnote{In the definition of $\MS(X,Z)$, for each $i$, for both $j=1,2$, the operators $X^i_j, Z^i_j$ act on $H_1 \otimes H_2$, but only nontrivially on $H_i$.  We do not write down the identity operator on $H_{\{1,2\}\setminus\{i\}}$ explicitly.  Both $X^1_1\otimes X^2_1 + Z^1_1\otimes Z^2_1$ and $X^1_2\otimes X^2_2 + Z^1_2\otimes Z^2_2$ are operators in $\Lin(\mH_1 \otimes \mH_2)$, and $\cdot$ denotes their product.  
The subscript $j$ can be interpreted as a label for systems within each player's local Hilbert space in the honest strategy.  
}$$
\item For any $\eps\geq 0$ there is a $\delta = O(\sqrt{\eps})$ such that
  for any strategy with success probability at least $1-\eps$, there
  are local isometries on $\mH_i$, $i\in\{1,2\}$, such that, under the
  isometries, $\ket{\psi} \simeq_\delta \ket{\EPR}_{\reg{A}_1\reg{B}_1}\otimes
  \ket{\EPR}_{\reg{A}_2\reg{B}_2}\otimes\ket{\psi'}$, for some state
  $\ket{\psi'}$.
  In addition, let $X_1$ and $X_2$ (resp.
  $Z_2$ and $Z_1$) be the binary observables associated with a
  player's first and second answer bits on question $r_1\in\mQ$ (resp.
  $r_2\in\mQ$).
  Then for $j \in \{1,2\}$ and $w\in\{x,z\}$, $W_j \simeq_\delta \sigma_{w,j}$, where
  $\sigma_{w,j}$ is the Pauli $\sigma_w$ observable acting on the player's
  $j$-th qubit.
  Similar approximations hold for questions $c_1$ and $c_2$, with the
  associated observables being close to $\sigma_{x,1}$ and $\sigma_{z,2}$, and
  $\sigma_{x,2}$ and $\sigma_{z,1}$ respectively.
\end{enumerate}
\end{theorem}

To derive the variant of the Magic Square game used in Theorem~\ref{thm:ms-rigid}, recall the standard formulation for the Magic Square as a matrix
\begin{equation}
  \label{eq:ms}
  \begin{bmatrix}
    xi & ix & xx \\
    iz & zi & zz \\
    xz & zx & yy
  \end{bmatrix}.
\end{equation}

In the formulation of the game from e.g.~\cite{aravind2004quantum}, the first player is sent a question which is a random entry in the matrix, and the second player a question which is a random row or column that contains the first player's question.  The first player replies with one bit and the second player replies with 3 bits.  The referee accepts if the answers are consistent, and the 3 answer bits of the second player multiply to $1$ except if her question is $c_3$ (the column with entries $xx,zz,yy$), in which case the product should be $-1$.  If the players share two $\EPR$ pairs, and measure the observables corresponding to the symbols in their questions (turning $ww'$ into the observable $\sigma_{w} \otimes \sigma_{w'}$) then they always succeed.

We consider the following modifications.
First, the questions we consider are always a complete row or a column, and never a single entry of the magic square.
This allows us to reduce the number of questions without changing the properties
of the game.
In addition, for simplicity we consider the uniform distribution on pairs of questions.
When the query to the players consists of two copies of the same row or column, the
referee checks that both answers from the players match.
If two non-insersecting rows or columns are sampled, then the referee automatically accepts.
Second, it is sufficient for the players to return the first $2$ bits
out of the $3$ she obtains, since the winning condition (known to the
player) forces the 3rd answer bit to be a deterministic function of
the first two.

Note that the soundness analysis (item 4.) only makes claims about the structure of a player's observables associated with the $x$ and $z$ symbols in the questions, not $y$. As explained in e.g.~\cite{reichardt2013classical}, due to the existence of two inequivalent non-trivial irreducible representations of the Pauli group (related by complex conjugation), this is inevitable.

\section{A 3-player rigidity test}
\label{sec:pauli}

In this section we introduce a 3-player test that can only be passed with probability $1$ by players who share two copies of
the $\GHZ$ state, and such that the players' observables associated to a subset of the questions in the test are isometric to $\sigma_x$ and $\sigma_z$ Pauli observables on their respective qubits. We obtain the test by combining the standard
$\GHZ$ test with the $\MS$ (Magic Square) test described in the previous section. The reason for using two $\GHZ$ states is that the Magic Square test requires two EPR pairs to be passed with probability $1$.

We call the resulting test the $\pauli$ test.
In this test, each player is asked to measure the two commuting
two-qubit Pauli operators that are indicated in the first two entries of the row or column of the magic square in
Eq.~\eqref{eq:ms} that she receives as her question, and return the outcomes as her answer.  
These answers are denoted by $a=(a_1,a_2)$, $b=(b_1,b_2)$, and $v=(v_1,v_2) \in \{\pm 1\}^2$ respectively.
The row or column sent to each player is chosen independently and uniformly at random by the referee among the $6$ possibilities.
The referee then checks all possible parity constraints implied by the
stabilizers of two $\GHZ$ states, among those that can be computed from the players' answers.
Suppose those two $\GHZ$ states lie on registers
$\reg{A}_1\reg{B}_1\reg{V}_1$ and
$\reg{A}_2\reg{B}_2\reg{V}_2$.
For example, if the query is $(r_1,r_1,r_1)$,
the measurement outcomes $a_1$, $a_2$
correspond to $xi$, $ix$ on $\reg{A}_1\reg{A}_2$,
$b_1$, $b_2$ correspond to $xi$, $ix$ on $\reg{B}_1\reg{B}_2$,
$v_1$, $v_2$ correspond to $xi$, $ix$ on $\reg{V}_1\reg{V}_2$,
so the referee checks that
$a_1b_1v_1=1$ and $a_2b_2v_2=1$, which corresponds to the stabilizers
$XIXIXI$ and $IXIXIX$ on $\reg{A}_1\reg{A}_2\reg{B}_1\reg{B}_2\reg{V}_1\reg{V}_2$.
If the query is $(r_1,r_3,c_3)$,
$a_1, a_2$ correspond to $xi, ix$ on $\reg{A}_1\reg{A}_2$,
$b_1, b_2$ correspond to $xz, zx$ on $\reg{B}_1\reg{B}_2$,
$v_1, v_2$ correspond to $xx, zz$ on $\reg{V}_1\reg{V}_2$,
so the referee
checks $a_1a_2b_1b_2v_1v_2=-1$, for the corresponding stabilizer
$XXYYYY$.

The complete test is described in Figure \ref{fig:3pauli}.
Intuitively, the $\pauli$ test embeds the Magic Square test as a three-player test, where two players in $\pauli$ jointly play the role of a single player in
the Magic Square game by measuring certain logical $X,Z$ observables.

\begin{figure}[htbp]
  \rule[1ex]{\textwidth}{0.5pt}\\
  Let $\mQ$ be the question set defined in~\eqref{eq:def-e-p}.
  The referee selects a query $(q_1,q_2,q_3)$ from $\mQ\times\mQ\times\mQ$
  uniformly at random, and sends one question to each player.
  Each player responds with two bits denoted by $a,b,v\in\{\pm 1\}^2$
  respectively.

For each query $q\in \mQ$, let $G_q$ be the group generated by the commuting two-qubit Pauli operators indicated in the corresponding row or column of the magic
  square described in Eq.~\eqref{eq:ms}.
  The group $G_q$ always contains four elements. Two of these elements are indexed by the first two entries in the row or column, and to these elements are associated the players' first two answers. In all cases except for the third column, the product of these elements is the third entry in the row or column, and to it the referee associates the product of the players' answers. For the case of the third column $c_3$, the value associated to the last square is the opposite of the product of the players' answers.

  If there is an operator $P \in G_{q_1}\times G_{q_2} \times G_{q_3}$ such
  that either $P$ or $-P$ is a stabilizer of the tensor product of two $\GHZ$ states,
  the referee rejects whenever the associated parity computed from the players' answers does not equal $+1$ or $-1$ respectively. In all other cases, the verifier accepts.
  \rule[1ex]{\textwidth}{0.5pt}
\caption{Description of the test $\pauli$.}
\label{fig:3pauli}
\end{figure}

\begin{theorem}[$3$-player Pauli test]\label{thm:3pauli}
There exists a three-player test $\pauli$, described in Figure~\ref{fig:3pauli}, with the following properties.
\begin{enumerate}
\item The marginal distribution on questions to each player is uniform
  over the set $\mQ$ defined in~\eqref{eq:def-e-p};
\item Each player replies with two bits in $\{\pm 1\}^2$;
\item For $i\in\{1,2,3\}$, let $\mH_i$ be a Hilbert space, and let $(X_1,Y_1,Z_1)$ and $(X_2,Y_2,Z_2)$ be any two commuting triples of observables satisfying the Pauli relations acting on $\mH_i$.  Then, there is a strategy for the players that succeeds with probability $1$ and only requires the measurement of observables obtained as the product of $X_1,Y_1,Z_1,X_2,Y_2,Z_2$ on an eigenvalue-$1$ eigenstate of the operator $G(X_1,Y_1)\cdot G(X_2,Y_2)$, where $G(\cdot,\cdot)$ is as in~\eqref{eq:ghz-op};
\item For any $\eps\geq 0$ there is $\delta = O({\eps}^{1/4})$ such that for
  any strategy with success probability at least $1-\eps$, there are
  local isometries on $\mH_i$, $i\in\{1,2,3\}$, such that, under the
  isometries, $\ket{\psi} \simeq_\delta
  \ket{\GHZ}_{123}\ket{\GHZ}_{123}\ket{\psi'}$ and for each player, the
  observables $X_1$, $X_2$ associated with the first and second answer
  bit to question $r_1$ and observables $Z_2$, $Z_1$ associated with
  the first and second answer bits to question $r_2$ satisfy $X_j\simeq_\delta
  \sigma_{x,j}$ and $Z_j \simeq_\delta \sigma_{z,j}$, where $\sigma_{w,j}$ is the Pauli
  $\sigma_w$ observable acting on the player's $j$-th qubit for
  $w\in\{x,z\}$.
\end{enumerate}
\end{theorem}

\begin{remark}
The soundness guarantees provided by Theorem~\ref{thm:3pauli} are analogous to those of Theorem~\ref{thm:ms-rigid}, except that they apply to a $3$-player test, two copies of the $\GHZ$ state, and the $\sigma_x$, $\sigma_z$ observables. The soundness parameter $\delta$ has a worse dependence on $\eps$, with an exponent $1/4$ instead of $1/2$. We did not try to optimize the exponent and expect it should be possible to improve it.
\end{remark}

\begin{proof}
  The first and second items are clear from the definition of the test
  in Figure~\ref{fig:3pauli}.
  To show the third, we describe a successful strategy for the
  players.
  Since $X_j,Y_j,Z_j$ satisfy the Pauli relations they are
  isomorphic to either the standard Pauli
  matrices~\eqref{eq:pauli-matrix} or their complex conjugate.
  For simplicity, assume the former, so that 
  each of the three players, upon receiving the question $q$, measures the first two commuting two-qubit Pauli operators in the corresponding
  row or column.
  As the referee only verifies the constraints implied by the
  stabilizer of the two $\GHZ$ states, it is immediate that the players succeed with
  probability $1$.

  We now show item 4., soundness.
  Consider a strategy for the players, using an arbitrary shared state
  $\ket{\psi}$ and projective measurements on that state, that succeeds
  with probability at least $1-\eps$.

  As the players' strategy uses projective measurement with four
  outcomes, they each define two observables.
  For each entry $m$ in the magic square, we define two observables,
  $R_m$ for the row and $C_m$ for the column. For example,
   row observables $R_{xi}$, $R_{ix}$ and $R_{xx}$ are derived from the four-outcome measurement applied by a player upon receiving
  question $r_1$. Similarly,  observables $C_{xx}$, $C_{zz}$ and $C_{yy}$
  are derived from the measurement applied by a player upon receiving question $c_3$.
  By definition, observables $R_m$ taken from the same row, or observables $C_m$
  taken from the same column, commute with each other.

  We show that for any of the nine possible values for $m$, the two observables $R_m$ and $C_m$ are close in the state-dependent distance.
  That is, the observables are almost identical, irrespective of whether the entry was asked as part of a row or a column.
  We show this for the example of the entry $m=xz$ asked to the first player; all other cases
  follow by a similar argument.

  The first step is to observe that it is always possible to select an element in the stabilizer group of $|\GHZ\>^{\otimes 2}$
  such that the first two tensor components (corresponding to the first qubit of each $\GHZ$ state) are $\sigma_{x,1}$ and $\sigma_{z,2}$. Here we can for example choose
  $\sigma_{x,1}\sigma_{z,2} \otimes \sigma_{x,1}\sigma_{z,2}\otimes \sigma_{x,1}\sigma_{i,2}$.
  In general, for any $m=rs$, we can find a stabilizer of $|\GHZ\>$ with the first tensor component being $\sigma_{r}$, since the stabilizer group of $|\GHZ\>$ includes $\sigma_x \otimes \sigma_x\otimes \sigma_x$ and $\sigma_z \otimes \sigma_z \otimes \sigma_i$, and similarly for $\sigma_{s}$.  Tensoring these two stabilizers gives a stabilizer of $|\GHZ\>^{\otimes 2}$ with $\sigma_{r,1}\sigma_{s,2}$ acting on the first qubits of the two $\GHZ$ states.


	As a second step, for precisely the same reason the chosen stabilizer can always be recovered from two distinct queries, the first with a row question involving $m$ to the first player, and the second with a column question involving $m$ to the first player. Here, the two queries are $(r_3,r_3,r_1)$ and $(c_1,r_3,r_1)$. For the case of $(r_3,r_3,r_1)$, one of the constraints verified by the referee is
  $a_1b_1v_1=1$, which implies that
  \begin{equation}\label{eq:r-1n}
    R_{xz} \otimes I \otimes I \approx_{\sqrt{\epsilon}} I \otimes R_{xz} \otimes R_{xi}.
  \end{equation}
  Similarly, for question $(c_1,r_3,r_1)$, the referee's check
  implies that
  \begin{equation*}
    C_{xz} \otimes I \otimes I \approx_{\sqrt{\epsilon}} I \otimes R_{xz} \otimes R_{xi}.
  \end{equation*}
  Combining the above two equations establishes that
  \begin{equation*}
    R_{xz} \approx_{\sqrt{\epsilon}} C_{xz}.
  \end{equation*}

Having shown analogous relations for each possible entry $m$ in the square, it follows that the set of operators $R_{m}$ approximately
  satisfies all the algebraic constraints for the operators in the
  magic square~\eqref{eq:ms}, i.e. they approximately multiply to the identity or its opposite for each row or column, as required. For example, $R_{xi}R_{ix}R_{xx}= I$ follows simply by definition of these observables. On the other hand, $R_{xx}R_{zz}R_{yy} \approx -I$ follows from the same using the column observables, which holds by definition, and the approximation $R_m \approx C_m$ shown above.

	It is then straightforward to devise a strategy for the two-player Magic Square test in which the first player determines her answers by measuring the observables $R_{m}$, and the role of the second player in the game is played by a joint strategy for the second and third players here, where each player measures the required observable that follows from using~\eqref{eq:r-1n} and analogous relations that hold for each possible entry. From the previous analysis it follows that the resulting strategy succeeds in
  the Magic Square test with probability  at least $1-O(\sqrt{\epsilon})$. Applying Theorem~\ref{thm:ms-rigid} it follows that there exists a local isometry such that
  \begin{equation*}
    R_{wi} \simeq_\delta \sigma_{w,1} \quad R_{iw} \simeq_\delta \sigma_{w,2},
  \end{equation*}
  for some $\delta = O(\eps^{1/4})$ and all $w\in{x,z}$.

To conclude, recall that by definition, $X_1 = R_{xi}$, $X_2 = R_{ix}$, $Z_1 = R_{zi}$ and
  $Z_2 = R_{iz}$.
  The characterization of the shared state claimed in item 4 follows
  from the form for $X_j$ and $Z_j$ described above, and the
  definition of the test, which in particular implies that the state is
  stabilized by $X_j\otimes X_j \otimes X_j$, $Z_j\otimes Z_j \otimes I$, and $I \otimes Z_j \otimes
  Z_j$, for $j\in\{1,2\}$.
\end{proof}

\section{Coherent state exchange with three players}
\label{sec:game}

In this section we describe our main result, a three-player game between a classical referee and three players that has the property that the optimal success probability of $1$ can only be achieved in the limit of arbitrarily high-dimensional entanglement. The first player in the game is called the ``virtual verifier'', $P_V$. The remaining two players are referred to using symbols $P_A$, $P_B$ respectively.
The game, called the $\game$ game, is described in Figure~\ref{fig:game-n}.
We first give some intuition behind the game.
%
%
In Section~\ref{sec:completeness} we exhibit a family of strategies for the players in the game, using states of growing dimension and with success probability that goes to $1$. In Section~\ref{sec:soundness} we show that any strategy for the players with success close to $1$ in the game must use an entangled state that has large local dimension.

\begin{figure}[htbp]
\rule[1ex]{\textwidth}{0.5pt}\\
The referee interacts with three players, labeled $P_V$, $P_A$ and
$P_B$. Each player receives a question taken from the set $\{0,1\} \times
\mQ$, where $\mQ$ is specified in~\eqref{eq:def-e-p}.
We use the symbol $\pi_{V}$, $\pi_A$, $\pi_B$ to denote the first component (lying in $\{0,1\}$) of the question to $P_V$, $P_A$ and $P_B$ respectively. It will always be the case that $\pi_A = \pi_B = \pi$.
In the game, $P_V$ should reply with 3 bits $(u,v) \in \{0,1\} \times \{\pm 1\}^2$, while
$P_A$, $P_B$ each reply with 2 bits $a,b \in \{\pm 1\}^2$ respectively.
Let $v=(v_1,v_2), a=(a_1,a_2), b=(b_1,b_2)$.
\\

The referee performs either of the following tests chosen at random with equal probability:
\begin{itemize}
\item[(a)] The referee sets $\pi_V = \pi = 0$. He executes the test $\pauli$ with the three players, inserting the question from $\pauli$ as the second component of their question, and checking validity of the triple $(v,a,b)$ extracted from the players' answers as would the verifier in $\pauli$.
\item[(b)] The referee sets $\pi_V=1$ and $\pi=0$.  The second component of $P_V$'s question is chosen uniformly at random from $\mQ$.  The referee performs either of the following with equal probability:
\begin{itemize}
\item[(i)] Send both $P_A$ and $P_B$ the question $r_2$.
  Let $a_1$ and $a_2$ be the answers associated with entries $iz$ and $zi$ respectively.
  Reject if $a_1=1$ and (($u=0$ and $a_2=-1$) or ($u=1$ and $a_2=1$)).
  Accept in all other cases.
\item[(ii)] Send both $P_A$ and $P_B$ the question $c_1$.
  Let $a_1$ and $a_2$ be the answers associated with entries $xi$ and $iz$ respectively.
  Reject if $a_2=-1$ and (($u=0$ and $a_1b_1=-1$) or ($u=1$ and
  $a_1b_1=1$)).
  Accept in all other cases.
\end{itemize}
\item[(c)] The referee sets $\pi_V=1$ and $\pi=0$. He sets the second component of $P_V$'s question to $r_2$. He sends both $P_A$ and $P_B$ the same question, $r_2$. The referee rejects if $a_1 \neq v_2$ or $b_1 \neq v_2$.
\item[(d)] The referee sets $\pi_V=\pi=1$, and executes the test $\pauli$ as in part (a). If $u=0$ the referee accepts if and only if the players' answers $(v,a,b)$ pass the test $\pauli$. If $u=1$ the referee always accepts.
\end{itemize}
\rule[1ex]{\textwidth}{0.5pt}
\caption{Description of the game $\game$.}
\label{fig:game-n}
\end{figure}

Before giving details of the analysis, we provide intuition behind the construction of the game. As in~\cite{leung2013coherent,regev2015quantum} the referee's goal in the game is to force $P_V$, $P_A$ and
$P_B$ to perform the transformation (normalization omitted)
\begin{equation}\label{eq:ideal-state}
  \begin{split}
    & \ket{0}_{\reg{V}} \; \ket{00}_{\reg{A}_1\reg{B}_1}\ket{00}_{\reg{A}_2\reg{B}_2} + \ket{1}_{\reg{V}} \; \ket{\EPR}_{\reg{A}_1\reg{B}_1} \ket{11}_{\reg{A}_2\reg{B}_2}\\
    \rightarrow\ & \ket{0}_{\reg{V}} \; \ket{00}_{\reg{A}_1\reg{A}_2}\ket{00}_{\reg{B}_1\reg{B}_2} + \ket{1}_{\reg{V}} \; \ket{11}_{\reg{A}_1\reg{B}_1} \ket{11}_{\reg{A}_2\reg{B}_2}\;.
  \end{split}
\end{equation}
Due to the EPR pair having one e-bit of entanglement, as opposed to the state $\ket{11}$ being a product state, the transformation~\eqref{eq:ideal-state} can be performed using operations local to $\reg{V}$, $\reg{A}$ and $\reg{B}$ \emph{only} by exploiting a large ancilla register that is used to ``embezzle'' the e-bit of entanglement.

The game $\game$ has two overlapping sub-games, indicated by a bit
$\pi_V \in \{0,1\}$ for $P_V$, and $\pi=\pi_A=\pi_B \in \{0,1\}$ for $P_A$
and $P_B$. The first sub-game, for $\pi_V=\pi=0$ (part (a) in
Figure~\ref{fig:game-n}), uses the test $\pauli$ to constrain the
players to share two copies of the $\GHZ$ state, on which they measure
$\sigma_x$ and $\sigma_z$ Pauli observables (embedded in questions in $\mQ$).

When $\pi_V=1$, player $P_V$ is tasked to perform a special measurement, which is obtained by applying a controlled-Hadamard from the qubit associated with his share of the second $\GHZ$ state, to the qubit associated with the first, followed by a measurement of the first qubit in the $\sigma_z$ basis. This yields the outcome labeled $u$. The goal of parts (b) and (c) of the game is to verify that $P_V$ applies precisely this measurement.

In the case $P_V$ obtained the measurement outcome $|0\>$ on the first qubit, it is a simple calculation (see Section~\ref{sec:completeness}) to verify that the three players share a state that is locally isometric to the state on the left-hand side of~\eqref{eq:ideal-state}.
 Now, observe that if the referee sometimes requires the three players to execute the test $\pauli$
on the second and third copies of the $\GHZ$ state, then conditioned
on $u=0$, in order to have a chance to succeed $P_A$ and $P_B$ have to
execute the transformation~\eqref{eq:ideal-state}, which brings the
second copy in a state that is locally isometric to a $\GHZ$ state. That they are able to achieve this is checked in part (d) of the game.

Note that the bits $\pi_V$ and $\pi$ are chosen so that $P_V$ can distinguish
part (a) from parts (b), (c) and (d), while $P_{A}$ and $P_{B}$ can distinguish parts (a), (b) and (c) from part (d). This allows the rigidity results obtained from the analysis of part (a) to carry
over to the analysis of parts (b) and (c): even though $P_V$ can distinguish those parts, $P_A$ and $P_B$ cannot, and $P_V$ cannot cheat on his own. But now if
$P_V$ plays parts (b) and (c) honestly, using the fact that he cannot distinguish parts (b), (c) and (d), $P_A$ and $P_B$ have to play part (d)
honestly as well.

\subsection{Completeness}
\label{sec:completeness}

We specify a sequence of strategies with growing dimension whose success probability approaches $1$.
The strategies follow closely the intuition for the game described earlier.

\begin{lemma}\label{lem:completeness}
For any integer $d\geq 1$ there exists a strategy for the players in $\game$ in which $P_V$ has three qubits and $P_A$ and $P_B$ each has $d+3$ qubits, such that the strategy is accepted with probability $1$ in parts (a), (b) and (c) of the game, and with probability $1-O(1/d)$ in part (d).
\end{lemma}

 For any integer $d\geq 1$, define an \emph{embezzlement state}
$$ \ket{\Gamma_d}_{\reg{A'B'}} = \frac{1}{\sqrt{N_d}} \sum_{j=1}^d \ket{11}_{\reg{A'B'}}^{\otimes j} \; \ket{\EPR}_{\reg{A'B'}}^{\otimes (d-j)}\;,$$
where $\ket{\EPR} = \frac{1}{\sqrt{2}}\ket{00} + \frac{1}{\sqrt{2}} \ket{11}$ and $N_d$ is a normalization constant such that $\|\ket{\Gamma_d}\|=1$. We think of $\ket{\Gamma_d}$ as a bipartite state on $\mH_{\reg{A'}} \otimes \mH_{\reg{B'}} \simeq (\C^2)^{\otimes d}_{\reg{A'}} \otimes (\C^2)^{\otimes d}_{\reg{B'}}$. This state has the property that there exists ``left-shift'' unitaries $W_{\reg{AA'}}$ on $\C^2_{\reg{A}} \otimes \mH_{\reg{A'}}$ and $W_{\reg{BB'}}$ on $\C^2_{\reg{B}} \otimes \mH_{\reg{B'}}$ such that
\begin{equation}\label{eq:emb-overlap}
 \big|\bra{11}_{\reg{AB}}\bra{\Gamma_d}_{\reg{A'B'}} \big(  W_{\reg{AA'}} \otimes W_{\reg{BB'}}\big) \ket{\EPR}_{\reg{AB}} \ket{\Gamma_d}_{\reg{A'B'}} \big| \,\geq\,1- O(1/d)\;.
\end{equation}

\begin{proof}[Proof of Lemma~\ref{lem:completeness}]
We define a strategy for the players in $\game$. The players share
\begin{equation}\label{eq:state}
\ket{\phi} \,=\, \ket{\GHZ}_{\reg{V}_1\reg{A}_1\reg{B}_1} \ket{\GHZ}_{\reg{V}_2\reg{A}_2\reg{B}_2} \ket{\GHZ}_{\reg{V}_3\reg{A}_3\reg{B}_3} \ket{\Gamma_d}_{\reg{A'B'}}\;.
\end{equation}
Here each of the registers $\reg{V}_j$, $\reg{A}_j$ and $\reg{B}_j$, for $j\in\{1,2,3\}$, is isomorphic to $\C^2$, and the registers $\reg{A'}$ and $\reg{B'}$ each has dimension $2^d$. Player $P_V$ holds registers $\reg{V}_1\reg{V}_2\reg{V}_3$, $P_A$ has $\reg{A}_1\reg{A}_2\reg{A}_3\reg{A'}$, and $P_B$ has $\reg{B}_1\reg{B}_2\reg{B}_3\reg{B'}$.

When $\pi_V = \pi = 0$, each player follows the honest strategy for $\pauli$ using her first and second qubits (item 3. in Theorem~\ref{thm:3pauli}).

If $\pi_V = 1$, $P_V$ performs a projective measurement $\Pi = \{\Pi^0,\Pi^1\}$ on his registers $\reg{V}_1\reg{V}_2$, where
\begin{equation}\label{eq:def-pi}
\Pi^0 =  \proj{0}_{\reg{V}_1}\otimes \proj{0}_{\reg{V}_2} +  \proj{+}_{\reg{V}_1}\otimes\proj{1}_{\reg{V}_2}\;,
~~~
\Pi^1 = \proj{1}_{\reg{V}_1}\otimes \proj{0}_{\reg{V}_2} + \proj{-}_{\reg{V}_1}\otimes \proj{1}_{\reg{V}_2}\;.
\end{equation}
The outcome determines her first answer bit $u\in\{0,1\}$. The player then applies a Hadamard on register $\reg{V}_1$, controlled on register $\reg{V}_2$.  (The measurement and the controlled-Hadamard have the same effect as a controlled-Hadamard followed by a measurement in the $\sigma_z$ eigenbasis on $\reg{V}_1$.)  We can already verify that this strategy succeeds with probability $1$ in part (b) of the test, which only depends on $P_V$'s answer $u$.  For $u=0$ and $1$, the post-measurement states of all players, after $P_V$ has applied the controlled-Hadamard, are
\begin{equation}\label{eq:phi0}
  \begin{split}
    \ket{\phi_0} = \ket{0}_{\reg{V}_1} \otimes \frac{1}{\sqrt{2}} \Big(\ket{0}_{\reg{V}_2} \ket{00}_{\reg{A}_1\reg{B}_1}\ket{00}_{\reg{A}_2\reg{B}_2} + & \ket{1}_{\reg{V}_2} \ket{\EPR}_{\reg{A}_1\reg{B}_1} \ket{11}_{\reg{A}_2\reg{B}_2} \Big) \otimes \\
    & \qquad \ket{\GHZ}_{\reg{V}_3\reg{A}_3\reg{B}_3}\otimes\ket{\Gamma_d}_{\reg{A'B'}}\;,
  \end{split}
\end{equation}
\begin{equation}\label{eq:phi1}
  \begin{split}
    \ket{\phi_1} = \ket{1}_{\reg{V}_1} \otimes \frac{1}{\sqrt{2}} \Big(\ket{0}_{\reg{V}_2} \ket{11}_{\reg{A}_1\reg{B}_1}\ket{00}_{\reg{A}_2\reg{B}_2} +  & \ket{1}_{\reg{V}_2} \ket{\EPR^-}_{\reg{A}_1\reg{B}_1} \ket{11}_{\reg{A}_2\reg{B}_2} \Big) \otimes \\
    & \qquad \ket{\GHZ}_{\reg{V}_3\reg{A}_3\reg{B}_3}\otimes\ket{\Gamma_d}_{\reg{A'B'}}\;,
  \end{split}
\end{equation}
where $\ket{\EPR^-} = \frac{1}{\sqrt{2}} \big( |00\> - |11\> \big)$.

In case (i), assuming $a_1=1$ $\left({\rm projecting~onto~}|0\>_{\reg{A}_2}\right)$ then $a_2=1$ $\left({\rm having~the~state~}|0\>_{\reg{A}_1}\right)$ with certainty and the referee accepts. In case (ii), assuming $a_2=-1$ $\left({\rm projecting~onto~}|1\>_{\reg{A}_2}\right)$ we have $a_1b_1=1$ with certainty (since $\sigma_x\otimes \sigma_x$ stabilizes $\ket{\EPR}$), so again the referee accepts.

To analyze part (c) we complete the description of $P_V$'s strategy when $\pi_V = 1$.
After the measurement to obtain $u$, $P_V$ takes the second part of the question and applies the honest strategy in game $\pauli$ using the appropriate Pauli operators on his registers $\reg{V}_2$ and $\reg{V}_3$.  It is then straightforward to verify that in both cases, $u=0$ or $u=1$, the players are accepted with certainty (note that by definition $P_V$'s answer $v_1$ is obtained by measuring $\sigma_z$ on register $\reg{V}_2$).

Note that $P_A$ and $P_B$ play parts (a), (b) and (c) using the same strategy (indeed, they have to, since they cannot distinguish questions coming from either of those parts of the game).

Finally we analyze part (d) (when $\pi_V=\pi=1$). First note that $P_V$ necessarily plays as already described in part (c).  Next we define a strategy for $P_A$ and $P_B$. Since in part (d) the referee always accepts in case $P_V$ reports $u=1$, it suffices to examine the players' strategy in case $u=0$. In this case, after $P_V$ has measured using $\Pi$ and applied the controlled-Hadamard, the post-measurement state of all players is as in~\eqref{eq:phi0}. Player $P_A$ (resp.\ $P_B$) performs a controlled-unitary $W_{\reg{A}_1\reg{A}'}$ (resp.\ $W_{\reg{B}_1\reg{B}'}$) as described in~\eqref{eq:emb-overlap}, controlled on the register $\reg{A}_2$ (resp.\ $\reg{B}_2$).
By~\eqref{eq:emb-overlap} the resulting state has overlap $1-O(1/d)$ with the state
\begin{align*}
\ket{0}_{\reg{V}_1} \otimes \frac{1}{\sqrt{2}} \Big( \ket{0}_{\reg{V}_2} \ket{00}_{\reg{A}_1\reg{A}_2}\ket{00}_{\reg{B}_1\reg{B}_2} + \ket{1}_{\reg{V}_2} \ket{11}_{\reg{A}_1\reg{B}_1} \ket{11}_{\reg{A}_2\reg{B}_2} \Big)  \otimes\ket{\GHZ}_{\reg{V}_3\reg{A}_3\reg{B}_3} \ket{\Gamma_d}_{\reg{A'B'}}\;.
\end{align*}
The player then applies a controlled-$\sigma_x$ operation on register $\reg{A}_1$ (resp.\ $\reg{B}_1$), controlled on $\reg{A}_2$ (resp.\ $\reg{B}_2$). This brings the state $O(1/d)$-close to
\begin{align}
\ket{000}_{\reg{V}_1\reg{A}_1\reg{B}_1} \otimes \frac{1}{\sqrt{2}} \Big( \ket{0}_{\reg{V}_2} \ket{0}_{\reg{A}_2}\ket{0}_{\reg{B}_2} + \ket{1}_{\reg{V}_2} \ket{1}_{\reg{A}_2} \ket{1}_{\reg{B}_2} \Big)  \otimes\ket{\GHZ}_{\reg{V}_3\reg{A}_3\reg{B}_3} \ket{\Gamma_d}_{\reg{A'B'}}\;.\label{eq:state-2}
\end{align}
At this point the player applies the honest strategy for the test $\pauli$ on the second and third copies of $|\GHZ\>$.
 Together with $P_V$'s strategy, due to the small discrepancy between the players' shared state and the ideal state in~\eqref{eq:state-2}, the players succeed with probability $1-O(1/d)$ in part (d).
\end{proof}

\subsection{Soundness}
\label{sec:soundness}

For the soundness analysis we rely on the following fact, implicit in~\cite[Section~3]{leung2013coherent} (building on results in \cite{van2003universal,fannes1973continuity}) and stated as Fact 5.7 in~\cite{regev2015quantum}.

\begin{fact}\label{fact:emb} Let $n,t$ be integers, $U,V\in\Lin(\C^n\otimes \C^t)$ arbitrary operators of norm at most $1$, and $\ket{\varphi}\in\C^n\otimes \C^n$, $\ket{\Psi}\in\C^t\otimes \C^t$ of unit norm. Let $S$ be the von Neumann entropy of the reduced density matrix of $\ket{\varphi}$ on any of the two subsystems, and assume $S\ge 1$. Then
$$1-\big|\bra{\varphi}\bra{\Psi} U\otimes V \ket{0^n 0^n}\ket{\Psi}\big|^2 \,\geq\, \min\Big\{ \frac{1}{4e^2},\,\frac{S^2}{16 \log^2 (3t)}\Big\}.$$
\end{fact}

We show the following.

\begin{lemma}\label{lem:soundness}
Suppose a strategy for the players succeeds with probability at least $1-\eps$ in the three-player game $\game$ described in Figure~\ref{fig:game-n}. Then the players must use an entangled state such that the local dimension of players $P_A$ and $P_B$ is at least $2^{\Omega(\eps^{-c})}$, for some constant $c>0$.
\end{lemma}


\begin{proof}
Fix a strategy for the players that succeeds with probability at least $1-\eps$ in the game. Let $\ket{\psi}_{\reg{VAB}} \in \mH_\reg{V} \otimes \mH_\reg{A} \otimes \mH_\reg{B}$ be the players' entangled state. 
We examine the consequences of the players' strategy having success probability at least $1-4\eps$ in each of the four parts of the game one after the other.

\paragraph{Part (a).}
Applying item 4.\ from Theorem~\ref{thm:3pauli}, for $D\in \{V,A,B\}$ there exists an isometry $W_D:\mH_\reg{D}\to \C^2_{\reg{D}_1} \otimes \C^2_{\reg{D}_2} \otimes \mH_{\reg{D}'}$ such that, under this isometry, the four-outcome POVM applied by a player to determine answers to a question of the form $w_1w_2\in\mP$ (when $\pi=0$) is isometric to the four-outcome POVM induced by Pauli $\sigma_{w_1}$ (when $w_1\in\{x,z\}$) and $\sigma_{w_2}$ (when $w_2\in\{x,z\}$) acting on $\reg{D}_1$ and $\reg{D}_2$ respectively, up to an error $\delta_1 = O({\eps}^{1/4})$.  Moreover, under all three isometries,
\begin{equation}\label{eq:r-0}
\ket{\psi}_{\reg{VAB}}\simeq_{\delta_1} \ket{\GHZ}_{\reg{V}_1\reg{A}_1\reg{B}_1}\ket{\GHZ}_{\reg{V}_2\reg{A}_2\reg{B}_2}\ket{\psi'}_{\reg{V'A'B'}}\;,
\end{equation}
for some tripartite state $\ket{\psi'}$.

For the remainder of the proof we modify the players' strategy to incorporate the isometry, and change their shared state to match exactly the state on the right-hand side of~\eqref{eq:r-0}; we keep the same notation $\ket{\psi}_{\reg{VAB}}$ for the modified state. The success probability of this modified strategy in parts (b), (c) and (d) of the game is at least $1-\eps_1$ for some $\eps_1 = O(\eps + \delta_1)=O(\delta_1) = O({\eps}^{1/4})$.

\paragraph{Part (b).}
When $\pi_V = 1$, $P_V$ applies an eight-outcome POVM measurement that we may assume to be projective. Let $w$ denote the second component of the question to $P_V$. For any value for $w$, let $\Pi_{w,\reg{V}} = \Pi_{w,\reg{V}}^0-\Pi_{w,\reg{V}}^1$ denote a binary observable associated with $P_V$'s first  answer bit, $u\in\{0,1\}$.

Consider the binary observable defined on $\C^2_{\reg{A}_1}\otimes \C^2_{\reg{A}_2}\otimes\C^2_{\reg{B}_1}\otimes\C^2_{\reg{B}_2}$ by
$$R_{\reg{AB}}\,=\,  \sigma_{z,\reg{A}_1} \otimes \proj{0}_{\reg{A}_2} + \sigma_{x,\reg{A}_1} \otimes \proj{1}_{\reg{A}_2}\otimes \sigma_{x,\reg{B}_1} \;.$$
Then we claim that for any $w$,
\begin{equation}\label{eq:r-1}
 \Pi_{w,\reg{V}} \otimes R_{\reg{AB}} \,\approx_{\sqrt{\eps_1}}\,I\;.
\end{equation}
To show~\eqref{eq:r-1} we decompose the $-1$ eigenspace of the observable $\Pi_{w,\reg{V}} \otimes R_{\reg{AB}} $ into a sum of two components, such that the overlap of each component with $\ket{\psi}$ can be bounded from the assumption that the strategy succeeds with probability $1-\eps_1$ in part (b). The first component,
$$ \Pi^1_{w,\reg{V}} \otimes \Big(\proj{0}_{\reg{A}_1} \otimes \proj{0}_{\reg{A}_2} +  \big(\proj{{+}{+}}_{\reg{A}_1\reg{B}_1} +   \proj{{-}{-}}_{\reg{A}_1\reg{B}_1}\big) \otimes \proj{1}_{\reg{A}_2}\Big)\;,$$
corresponds to rejection for the $u=1$ cases of (i) and (ii) in part (b). The second component,
$$ \Pi^0_{w,\reg{V}} \otimes \Big(\proj{1}_{\reg{A}_1} \otimes \proj{0}_{\reg{A}_2} +  \big(\proj{{+}{-}}_{\reg{A}_1\reg{B}_1} +   \proj{{-}{+}}_{\reg{A}_1\reg{B}_1}\big) \otimes \proj{1}_{\reg{A}_2}\Big)\;,$$
corresponds to rejection for the $u=0$ cases. This shows~\eqref{eq:r-1}. Let $\Pi'_{\reg{V}} = \sigma_{z,\reg{V}_1} \otimes \proj{0}_{\reg{V}_2} + \sigma_{x,\reg{V}_{1}} \otimes \proj{1}_{\reg{V}_2}$. Using the fact that the $\GHZ$ state is stabilized by $\sigma_{z,\reg{V}_1} \otimes \sigma_{z,\reg{A}_1}$ as well as by $\sigma_{x,\reg{V}_1} \otimes \sigma_{x,\reg{A}_1} \otimes \sigma_{x,\reg{B}_1}$, it follows  that $\Pi'_{\reg{V}} \otimes R_{\reg{AB}} \ket{\psi} = \ket{\psi}$. Together with~\eqref{eq:r-1}, we have shown that $\Pi_{w,\reg{V}} \approx_{\eps_1} \Pi'_{\reg{V}}$ for all $w$; in particular $\Pi_{w,\reg{V}}$ does not depend on $w$ (to the extent that only its action on $\ket{\psi}$ is considered), and for the remainder of the proof we drop the subscript $w$.

\paragraph{Part (c).} Let $Z_1$ be the observable associated with $P_V$'s outcome $v_2 \in \{\pm 1\}$ when the second component of his question is $r_2$. Using the fact that $\ket{\GHZ}_{\reg{V}_2\reg{A}_2\reg{B}_2}$ is stabilized by $\sigma_{z}\otimes \sigma_z$ acting on $\reg{V}_2\reg{A}_2$ or $\reg{V}_2\reg{B}_2$, success $1-\eps_1$ in this part enforces that
\begin{equation}\label{eq:r-2}
Z_1 \, \approx_{\sqrt{\eps_1}} \,\sigma_{z,\reg{V}_2}\;.
\end{equation}

\paragraph{Part (d).}
From the analysis of part (b) we deduce that conditioned on the referee choosing to execute part (d), and on the outcome $u=0$ having been obtained from $P_V$, the joint state of the players (irrespective of the choice of question $w$ to $P_V$) is
$$ \ket{\psi''} \,\simeq_{\delta_2}\,  \frac{1}{\sqrt{2}} \big(\ket{00}_{\reg{V}_1\reg{V}_2} \ket{00}_{\reg{A}_1\reg{B}_1}\ket{00}_{\reg{A}_2\reg{B}_2} + \ket{+}_{\reg{V}_1}\ket{1}_{\reg{V}_2} \ket{\EPR}_{\reg{A}_1\reg{B}_1} \ket{11}_{\reg{A}_2\reg{B}_2} \big) \otimes \ket{\psi'}_{\reg{VA'B'}}\;,$$
for some $\delta_2=O(\sqrt{\eps_1})= O({\eps}^{1/8})$.
From part (c) we also know that $P_V$'s observable $Z_1$ associated
with answer $v_2$ to question $r_2$ satisfies~\eqref{eq:r-2}.
By item 4.\
from Theorem~\ref{thm:3pauli}, for the players' strategy to succeed
with probability $1-\eps_1$ in part (d) it is necessary that the
observable $X_{1}$ associated with $P_V$'s answer $v_1$ on question
$r_1$ approximately anti-commutes with $Z_1$ on $\ket{\psi''}$, up to
an error $\eps_2 = O(\eps_1^{1/4})= O({\eps}^{1/16})$.
Using~\eqref{eq:r-2}, it follows that $X_{1} \approx_{\eps_2}
\ket{0}\bra{1}_{\reg{V}_2} \otimes U + \ket{1}\bra{0}_{\reg{V}_2} \otimes U^\dagger$,
for some unitary $U$ on $\reg{V}_1\reg{V'}$.
Using again item 4.\
of Theorem~\ref{thm:3pauli}, success in part (d) also implies that
$X_1 \otimes X_\reg{1,A} \otimes X_{1,\reg{B}}$ approximately stabilizes
$\ket{\psi''}$, where $X_\reg{1,A}$ and $X_\reg{1,B}$ are observables
associated with $P_A$ and $P_B$'s outcome $a_1$ and $b_1$ on question
$x_1i_2$ respectively.
Thus
\begin{equation*}
  \begin{split}
    \Big| \bra{0}_{\reg{V}_1} \bra{\psi'}_{\reg{V'A'B'}} \bra{11}_{\reg{A}_2 \reg{B}_2} \bra{\EPR}_{\reg{A}_1\reg{B}_1} \big( U_{\reg{V}_1\reg{V'}} & \otimes X_{1,\reg{A}} \otimes X_{1,\reg{B}} \big)\\
    & \qquad \ket{00}_{\reg{A}_1\reg{B}_1}\ket{00}_{\reg{A}_2\reg{B}_2} \ket{\psi'}_{\reg{V'A'B'}}\ket{+}_{\reg{V}_1}\Big| \geq 1-\delta_3\;,
  \end{split}
\end{equation*}
for some $\delta_3 = O(\delta_2+\eps_2)= O({\eps}^{1/16})$.
 Expanding out registers $\reg{V}_1\reg{V'}$ of $\ket{+}_{\reg{V}_1}\ket{\psi'}_{\reg{V'A'B'}}$ in the eigenbasis of $U$, we obtain a distribution $\{|\alpha_k|^2\}_k$ and a family of states $\{\ket{\psi^{(k)}}_{\reg{A'B'}}\}_k$ such that
$$ \sum_k |\alpha_k|^2 \Big| \bra{\psi^{(k)}}_{\reg{A'B'}} \bra{11}_{\reg{A}_2 \reg{B}_2} \bra{\EPR}_{\reg{A}_1\reg{B}_1} \big(X_{1,\reg{A}} \otimes X_{1,\reg{B}} \big) \ket{00}_{\reg{A}_1\reg{B}_1}\ket{00}_{\reg{A}_2\reg{B}_2} \ket{\psi^{(k)}}_{\reg{A'B'}}\Big| \geq 1-O(\delta_3)\;.$$
By Fact~\ref{fact:emb} applied with $n=2$ and $\ket{\varphi} = \ket{\EPR}$ we deduce the claimed lower bound on the dimension of the players' strategy.  In particular, 
${\frac{1}{16 \log^2(3t)} \leq O(\delta_3)} = O({\eps}^{1/16})$, so, $t = 2^{\Omega({\eps}^{-1/32})}$. 
\end{proof}

\section{Acknowledgements}

We thank William Slofstra and John Watrous for helpful discussions.
Part of this work was conducted when ZJ was visiting the Institute for Quantum
Computing in the University of Waterloo, and also when DL and TV were attending
the ``Quantum Physics of Information'' program at the Kavli Institute for
Theoretical Physics.
This research was supported in part by the Australian Research Council under
Grant No.\ DP200100950 and the National Science Foundation under Grant No.\ NSF
PHY 17-48958.
DL is supported by an NSERC Discovery grant and a CIFAR research grant via the
Quantum Information Science program.
TV is supported by NSF CAREER Grant CCF-1553477, AFOSR YIP award number
FA9550-16-1-0495, a CIFAR Azrieli Global Scholar award, and the IQIM, an NSF
Physics Frontiers Center (NSF Grant PHY-1125565) with support of the Gordon and
Betty Moore Foundation (GBMF-12500028).

%
%
\bibliographystyle{quantum-abbrv}
\bibliography{embezzle}


\end{document}